\documentclass[letterpaper, 10 pt, conference]{ieeeconf}  % Comment this line out
\IEEEoverridecommandlockouts                              % This command is only
\usepackage{amsmath,amssymb,enumerate,subfigure,multirow,hhline,color}
\usepackage{xcolor}
\usepackage{hyperref}
\usepackage{graphicx}
\usepackage{graphics}
\usepackage[sort]{cite}

\usepackage{algorithm}
\usepackage{algpseudocode}

\usepackage{epsfig,psfrag}
\usepackage{tabularx}
\usepackage{tabulary}

\usepackage[sort]{cite}

\usepackage{hyperref}
\hypersetup{breaklinks=true,colorlinks=true,linkcolor=black,citecolor=black}

\usepackage[noabbrev,capitalise,nameinlink]{cleveref}

\DeclareMathOperator{\argmin}{arg\,min}

\newtheorem{example}{Example}
\newtheorem{assumption}{Assumption}

\newtheorem{problem}{Problem}

\newtheorem{proposition}{Proposition}

\allowdisplaybreaks

\title{\LARGE \bf
Continuous-Time Distributed Dynamic Programming for Networked Multi-Agent Markov Decision Processes
}

\author{Donghwan Lee, Han-Dong Lim, and Do Wan Kim
\thanks{D. Lee and H. Lim are with the Department of Electrical and Engineering, Korea Advanced Institute of Science and Technology (KAIST), Daejeon, 34141, South Korea {\tt\small
donghwan@kaist.ac.kr, limaries30@kaist.ac.kr}.}%}
\thanks{D. Kim is with the Department of Electrical Engineering,
Hanbat National University, Daejeon 34158, South Korea {\tt\small
dowankim@hanbat.ac.kr}.}
\thanks{$^*$ This work was supported by National Research Foundation under Grant NRF-2021R1I1A3058581 and Institute of Information communications Technology Planning Evaluation (IITP) grant funded by the Korea government (MSIT)(No.2022-0-00469)}
}

\begin{document}

\maketitle
\thispagestyle{empty}
\pagestyle{empty}

%%%%%%%%%%%%%%%%%%%%%%%%%%%%%%%%%%%%%%%%%%%%%%%%%%%%%%%%%%%%%%%%%%%%%%%%%%%%%%%%
\begin{abstract}
The main goal of this paper is to investigate continuous-time distributed dynamic programming (DP) algorithms for networked multi-agent Markov decision problems (MAMDPs). In our study, we adopt a distributed multi-agent framework where individual agents have access only to their own rewards, lacking insights into the rewards of other agents. Moreover, each agent has the ability to share its parameters with neighboring agents through a communication network, represented by a graph. We first introduce a novel distributed DP, inspired by the distributed optimization method of Wang and Elia. Next, a new distributed DP is introduced through a decoupling process. The convergence of the DP algorithms is proved through systems and control perspectives.
The study in this paper sets the stage for new distributed temporal different learning algorithms.
\end{abstract}

\section{Introduction}

A Markov decision problem (MDP)\cite{sutton1998reinforcement,puterman1990markov} is a sequential decision-making problem that aims to find an optimal policy in dynamic environments. Multi-agent MDPs (MAMDPs)\cite{zhang2021multi,lee2020optimization} extend the MDP framework to include multiple agents interacting with one another. These agents can either cooperate toward a shared goal or compete for individual objectives. In this study, we focus primarily on the cooperative scenario.

In MAMDPs, full information about the environment, such as the global state, action, and reward, is often unavailable to each agent. This lack of complete information can arise for a variety of reasons, including sensor or infrastructure limitations, privacy and security concerns, and computational constraints, among others. As a result, various information structures are adopted based on the specific application. One notable instance is the centralized MAMDP, where every agent has access to complete information. In contrast, in distributed MAMDPs, agents may only have access to local data about the global state, action, and rewards. Sometimes, agents can share information with each other via communication networks, an environment termed as the networked MAMDP.

In this paper, we investigate new continuous-time distributed DP algorithms for a networked MAMDP. In this setting, agents can share their local parameters with their neighbors over a communication network described by a graph, $\mathcal{G}$. The proposed algorithms are distributed in the sense that only local rewards, $r_i$, are given to each agent. Meanwhile, the global reward is a sum or an average of the local rewards, i.e., $r = (r_1 + r_2 + \cdots + r_N)/N$, where $N$ is the total number of agents. To address the MAMDP in a distributed manner, we employ the distributed optimization techniques~\cite{jadbabaie2003coordination,nedic2009subgradient,nedic2010constrained,shi2015extra}. These techniques enable multiple agents to calculate a common solution through parameter mixing (or averaging) steps with their neighbors.

In particular, we introduce two novel continuous-time distributed DP~\cite{bertsekas1996neuro} algorithms. The first algorithm is inspired by the distributed optimization technique of Wang and Elia~\cite{wang2010control}. The second DP algorithm is developed through a special decoupling process. The convergence of these algorithms is proved from systems and control perspectives~\cite{khalil2002nonlinear}.

The main contributions of this paper are summarized as follows:
To the authors' best knowledge, the algorithms presented in this paper are the first attempts to develop and analyze distributed DP algorithms characterized by simple continuous-time linear dynamics. These algorithms are readily analyzable from a control theory standpoint. This approach, based on systems and control theory, simplifies and clarifies the analysis, especially for those with a background in control theory, and provides additional insights into the distributed DP. Moreover, this paper establishes a foundation for developing new distributed temporal difference learning algorithms.

While numerous studies~\cite{macua2015distributed,stankovic2016multi,sha2020asynchronous,doan2019finite,wai2018multi,cassano2020multiagent,ding2019fast,heredia2020finite,lee2022distributed} have studied model-free distributed temporal difference learning algorithms under a variety of scenarios and conditions, this paper is among the first to thoroughly investigate model-based continuous-time linear dynamics. Although model-free methods are generally more applicable in a broader range of situations, the approaches in this paper can be readily extended to model-free temporal difference learning techniques.

\section{Preliminaries}
\subsection{Notation and terminology}
The following notation is adopted: ${\mathbb R}^n $ denotes the $n$-dimensional Euclidean space; ${\mathbb R}^{n \times m}$ denotes the set
of all $n \times m$ real matrices; ${\mathbb R}_+ $ and ${\mathbb
R}_{++}$ denote the sets of nonnegative and positive real numbers,
respectively, $A^T$ denotes the transpose of matrix $A$; $I_n$ denotes the $n \times n$
identity matrix; $I$ denotes the identity matrix with  appropriate
dimension; $\|\cdot \|_2$ denotes the standard Euclidean norm; $\|x\|_D:=\sqrt{x^T Dx}$ for any
positive-definite $D$; $\lambda_{\min}(A)$ denotes the minimum eigenvalue of
$A$ for any symmetric
matrix $A$; $|{\mathcal S}|$ denotes
the cardinality of the set for any finite set ${\mathcal S}$; $[x]_i$ is the $i$-th element for any
vector $x$; $[P]_{ij}$ indicates the element in $i$-th row and $j$-th column for any matrix $P$;
if ${\bf z}$ is a discrete random variable which has $n$ values
and $\mu \in {\mathbb R}^n$ is a stochastic vector, then ${\bf z}
\sim \mu$ stands for ${\mathbb P}[{\bf z} = i] = [\mu]_i$ for all
$i \in \{1,\ldots,n \}$; ${\bf 1}_n \in {\mathbb R}^n$ denotes an
$n$-dimensional vector with all entries equal to one.

\subsection{Graph theory}
An undirected graph with the node set ${\mathcal V}$ and the edge set
${\mathcal E}\subseteq {\mathcal V}\times {\mathcal V}$ is denoted by ${\mathcal
G}=({\mathcal E},{\mathcal V})$. We define the neighbor set of node $i$ as
${\mathcal N}_i :=\{j\in {\mathcal V}:(i,j)\in {\mathcal E}\}$. The adjacency
matrix of $\mathcal G$ is defined as a matrix $W$ with $[W]_{ij} = 1$,
if and only if $(i,j) \in {\mathcal E}$. If $\mathcal G$ is undirected,
then $W=W^T$. A graph is connected, if there is a path between any
pair of vertices. The graph Laplacian is $L=H - W$, where $H$ is a
diagonal matrix with $[H]_{ii} = |{\mathcal N}_i|$. If the graph is
undirected, then $L$ is symmetric positive semi-definite. It holds
that $L {\bf 1}_{|{\mathcal V}|}=0$. If ${\mathcal G}$ is connected, $0$
is a simple eigenvalue of $L$, i.e., ${\bf 1}_{|{\mathcal V}|}$ is the unique eigenvector corresponding to $0$, and the span of ${\bf 1}_{|{\mathcal V}|}$ is the null space of $L$.
In this paper, we assume that the underlying network is connected.
\begin{assumption}\label{assumption:connected}
$\cal G$ is connected.
\end{assumption}

\subsection{Markov decision process}\label{sec:MDP-BE}
A Markov decision process (MDP)~\cite{sutton1998reinforcement} is characterized by a quadruple ${\mathcal M}: =
({\mathcal S},{\mathcal A},P,r,\gamma)$, where ${\mathcal S}$ is a finite
state space (observations in general), $\mathcal A$ is a finite action
space, $P(s,a,s')$ represents the (unknown)
state transition probability from state $s$ to $s'$ given action
$a$, $r:{\mathcal S}\times {\mathcal A}\times {\mathcal S}$ is the reward
function, and $\gamma \in (0,1)$ is the discount factor. In particular, if action
$a$ is selected with the current state $s$, then the state
transits to $s'$ with probability $P(s,a,s')$ and incurs a random
reward $r(s,a,s')$. The stochastic policy is a map $\pi:{\mathcal S} \times
{\mathcal A}\to [0,1]$ representing the probability $\pi(a|s)$ of taking action $a$ at the current state $s$, $P^\pi$ denotes the transition matrix, and $d:{\mathcal S} \to {\mathbb R}$ denotes the stationary
distribution of the state $s\in {\mathcal S}$ under the policy $\pi$. We also define
$R^\pi(s)$ as the expected reward given the policy $\pi$ and the current state $s$. The infinite-horizon discounted value function with policy $\pi$ is
\begin{align*}
%================================================================
&J^\pi(s):={\mathbb E} \left[ \left. \sum_{k = 0}^\infty {\gamma
^k r(s_k,a_k,s_{k+1})} \right|s_0=s \right],
%================================================================
\end{align*}
where ${\mathbb E}$ stands for the expectation
taken with respect to the state-action trajectories following the
state transition $P^\pi$. Given pre-selected basis (or feature)
functions $\phi_1,\ldots,\phi_q:{\mathcal S}\to {\mathbb R}$, $\Phi
\in {\mathbb R}^{|{\mathcal S}| \times q}$ is defined as a full column
rank matrix whose $s$-th row vector is $\phi(s):=\begin{bmatrix} \phi_1(s)
&\cdots & \phi_q(s) \end{bmatrix}$. The goal of the Markov decision problem with the linear
function approximation is to find the weight vector $\theta$ such that
$J_{\theta}:=\Phi \theta$ approximates the true value function $J^{\pi}$.
This is typically done by minimizing the {\em mean-square projected Bellman
error} loss function~\cite{sutton2009fast}
\begin{align}
&{\min _{\theta  \in {\mathbb R}^q}}{\rm{MSPBE}}(\theta ): = {\min _{\theta  \in {\mathbb R}^q}}\frac{1}{2}\left\| {\Pi ({R^\pi } + \gamma {P^\pi }\Phi \theta  - \Phi \theta )} \right\|_D^2,\label{eq:mspbe}
\end{align}
where $D$ is a diagonal matrix with positive diagonal elements $d(s),s\in {\mathcal S}$, $\Pi$ is the projection onto the range space of $\Phi$, denoted by $R(\Phi)$: $\Pi(x):=\argmin_{x} \|x-x'\|_D^2$, $x'\in R(\Phi)$, and $R^\pi \in {\mathbb R}^{|{\mathcal S}|}$ is a vector enumerating all $R^\pi(s), s\in {\mathcal S}$. The projection can be performed by the matrix multiplication: we write $\Pi(x):=\Pi x$, where $\Pi:=\Phi(\Phi^T D\Phi)^{-1}\Phi^T D$.
The solutions of~\eqref{eq:mspbe} is known to be equivalent to those of the so-called projected Bellman equation
\begin{align}
\Phi\theta = \Pi(R^\pi+\gamma P^\pi \Phi\theta),\label{eq:projected-Bellman}
\end{align}
whose solution is given by
\begin{align}
\theta^* =-(\Phi^T D(\gamma P^\pi-I)\Phi)^{-1}\Phi^TDR^\pi. \label{eq:sol-of-MSPBE}
\end{align}

\section{Multi-agent MDP}\label{section:distributed MAMDP-overview}
In this section, we introduce the notion of the
distributed MAMDP, which will be studied throughout the paper.
Consider $N$ agents labeled by $i \in \{ 1,\ldots,N\}=:{\mathcal
V}$. A multi-agent Markov decision process is characterized by
$({\mathcal S},\{ {\mathcal A}_i\}_{i\in {\mathcal
V}},P,\{r_i\}_{i \in {\mathcal V}},\gamma)$, where $\gamma \in
(0,1)$ is the discount factor, ${\mathcal S}$ is a finite state
space, ${\mathcal A}_i$ is a finite action space of agent
$i$, $a:=(a_1,\ldots,a_N)$ is the joint
action, ${\mathcal A}:=\prod_{i=1}^N {{\mathcal
A}_i}$ is the
corresponding joint action space, $r_i:{\mathcal S} \times {\mathcal A} \times {\mathcal S} \to {\mathbb R}$
is a reward function of agent $i$, and $P(s,a,s')$ represents the transition model of the state $s$ with the
joint action $a$ and the corresponding joint action space ${\mathcal
A}$. The stochastic policy of agent $i$ is a mapping $\pi_i:{\mathcal
S} \times {\mathcal A}_i \to [0,1]$ representing the probability
$\pi_i(a_i|s)$ of selecting action $a_i$ at the state $s$, and the corresponding joint
policy is $\pi(a|s):=\prod_{i=1}^N {\pi_ i(a_i|s)}$. Moreover, $P^{\pi}$
denotes the transition matrix, and $d:{\mathcal S}
\to {\mathbb R}$ denotes the stationary state distribution under the joint policy $\pi$.
In particular, if the joint action $a$ is
selected with the current state $s$, then the state transits to
$s'$ with probability $P(s,a,s')$, and each agent $i$ observes a
reward $r_i(s,a,s')$. In addition, $J^\pi$ is the infinite-horizon discounted value function
with policy $\pi$ and reward $r =(r_1 +\cdots+ r_N
)/N$ satisfying $J^\pi=\frac{1}{N}\sum_{i =
1}^N {R_i^{\pi}}+\gamma P^\pi J^\pi$.

\begin{problem}[Distributed value evaluation problem]\label{problem:multi-agent-prob}
The goal of each agent $i$ is to find the value function of the centralized reward $r=(r_1+\cdots+ r_N)/N$, where only the local reward $r_i$ is given to each agent, and parameters can be shared with its neighbors over communication network represented by the graph $\cal G$.
\end{problem}

We also note that we can also consider the following scenario: there are $N$ agents behave in $N$ copies of identical and independent environments, and each agent $i$ observes the current state $s$ in its own environment, executes an action $a \in {\cal A}$ according to the policy $\pi$, and it causes the state $s\in {\cal S}$ to transit to $s' \in {\cal S}$ with probability $P(s,a,s')$ in each independent environment. Then, the agent receives the local reward $r_i(s,a,s')$.

In this paper, we assume that each agent does not have access to the rewards of the other agents. For instance, there is no centralized coordinator; thus, each agent is unaware of the rewards of other agents. On the other hand, we suppose that each agent knows only the parameters of adjacent agents over the network graph, assuming that the agents can communicate with each other. Without each agent knowing the full reward algorithm of the group, our algorithm produces the same result as if each agent were receiving the average rewards of the group.

\section{Continuous-time distributed dynamic programming}
For the sake of notational simplicity in representing a multi-agent environment, we first introduce the stacked vector and matrix notations:
\begin{align*}
&\bar \theta: = \begin{bmatrix}
   \theta_1 \\
    \vdots \\
   \theta_N \\
\end{bmatrix},\quad \bar R^\pi  := \begin{bmatrix}
   R_1^\pi \\
    \vdots \\
   R_N^\pi \\
\end{bmatrix},\\
&\bar P^{\pi}:=I_N \otimes P^{\pi},\quad \bar L:=L
\otimes I_q,\quad \bar D:=I_N \otimes D,\\
& \bar\Phi:=I_N \otimes \Phi
\end{align*}
Before delving into the distributed dynamic programming (DP), it is beneficial to examine the centralized version of DP. This provides a foundation that can be extended to the distributed version. The centralized variant can be naturally derived from the solution of the MSPBE for a single agent, as shown in~\eqref{eq:sol-of-MSPBE}.

\subsection{Centralized dynamic programming}
In the centralized multi-agent case, the same reward $R_c^\pi =(R_1^\pi+ \cdots + R_N^\pi)/N$ for every agent is given. Then, it can be simply considered as the single agent case with stacked vector and matrix notation. According to the single-agent MDP results in~\eqref{eq:sol-of-MSPBE}, the optimal solution is given as
\begin{align}
{\bar \theta ^*} :=  - {({\bar \Phi ^T}\bar D(\gamma {\bar P^\pi } - I)\bar \Phi )^{ - 1}}{\bar \Phi ^T}\bar D({\bf{1}}_N \otimes R_c^\pi ),\label{eq:theta-global2}
\end{align}
which minimizes the corresponding MSPBE~\eqref{eq:mspbe}.
Using algebraic manipulations, we can easily prove that $\bar \theta ^*$ can be represented by $\bar \theta ^*  = {\bf{1}}_N \otimes \theta _\infty ^ c$, where
\begin{align}
\theta_\infty^c := - (\Phi^T D(-I + \gamma P^\pi)\Phi)^{-1} \Phi^T D R_c^\pi.\label{eq:theta-global}
\end{align}

The solution can be found using a standard DP method~\cite{bertsekas1996neuro}. In this paper, we will consider a DP in the continuous-time domain (or ODE) as follows:
\begin{align}
\frac{d}{dt}\bar\theta_t =\bar\Phi^T \bar D(-I+\gamma\bar P^\pi)\bar\Phi\bar\theta_t+\bar\Phi^T \bar D({\bf{1}}_N\otimes R_c^\pi),\label{eq:1}
\end{align}
which is a linear ODE. We can easily prove that $\bar \theta ^*$ is an asymptotically stable equilibrium point.
\begin{proposition}~\label{thm:proposition-central}
$\bar \theta ^*$ is a unique asymptotically stable equilibrium point of the linear system in~\eqref{eq:1}, i.e., $\bar \theta_t \to \bar \theta^*$ as $t\to \infty$.
\end{proposition}
\begin{proof}
Using $\bar \Phi ^T \bar D(\gamma \bar P^\pi   - I)\bar \Phi \bar \theta ^*  + \bar \Phi ^T \bar D\bar R^\pi   = 0$,~\eqref{eq:1} can be represented by the linear system
\[
\frac{d}{dt}(\bar \theta _t  - \bar \theta ^* ) = \bar \Phi ^T \bar D( - I + \gamma \bar P^\pi  )\bar \Phi (\bar \theta _t  - \bar \theta ^*).
\]
We can easily prove that $\bar \Phi ^T \bar D( - I + \gamma \bar P^\pi  )\bar \Phi$ is negative definite, and hence, Hurwitz~\cite[pp.~209]{bhatnagar2012stochastic}. Therefore, $\bar \theta _t  - \bar \theta ^* \to 0$ as $t \to \infty$, which completes the proof.
\end{proof}

To solve~\eqref{eq:1}, we must assume that the central reward $R_c^\pi$ is accessible to all agents. In subsequent sections, we will explore distributed versions where only the local reward $R^\pi_i$ is provided to each agent $i$. We present two versions: the first is inspired by~\cite{wang2010control}, while the second is a novel approach that offers more desirable properties compared to the first when integrated into reinforcement learning (RL) frameworks~\cite{sutton1988learning}.

\subsection{Distributed dynamic programming version~1}
In the networked multi-agent setting, each agent receives each of their local rewards, and parameters from neighbors over a communication graph. Based on the ideas of Wang and Elia in~\cite{wang2010control}, we can convert the continuous-time ODE in~\eqref{eq:1} into
\begin{align}
\frac{d}{dt}\bar\theta_t &= \bar\Phi^T \bar D(-I+\gamma \bar P^\pi)\bar\Phi\bar\theta_t + \bar\Phi^T \bar D\bar R^\pi-\bar L\bar\theta_t-\bar L\bar w_t,\nonumber\\
\frac{d}{dt}\bar w_t  &= \bar L\bar \theta_t.\label{eq:ver1}
\end{align}
Compared to~\eqref{eq:1}, the above ODE consists of an auxiliary vector $\bar w_t$ and the graph Laplacian matrix $\bar L$.
Here, the Laplacian helps the consensus of each agent, and the auxiliary vector potentially allows agents make better use of their local information. Note that each agent only uses local information by multiplying the Laplacian in both equations. The ODE in~\eqref{eq:ver1} can be written as~\Cref{algo:1} from a local view.
%=====================================================================================
\begin{algorithm}[t]

\caption{Distributed dynamic programming version~1}\label{algo:1}
%=====================================================================================
\begin{algorithmic}[1]
\State Initialize $\theta_0^i,w_0^i,i\in \{1,2,\ldots,N \}$.
\For{$t \ge 0$}
   \For{agent $i\in \{1,\ldots,N\}$}
    \State Update
    \begin{align*}
    \frac{d}{dt}\theta_t^i &= \Phi^T D(-I + \gamma P^\pi)\Phi\theta_t^i + \Phi^T DR_i^\pi\\
    &-\left(|{\cal N}_i |{\theta_t^i} - \sum_{j\in {\cal N}_i} {\theta_t^j} \right) - \left( |{\cal N}_i |w_t^i - \sum_{j\in {\cal N}_i } {w_t^j} \right)\\
    \frac{d}{dt}w_t^i &= |{\cal N}_i|\theta_t^i - \sum_{j\in {\cal N}_i} {\theta_t^j}
    \end{align*}
    where ${\mathcal N}_i$ is the neighborhood of node $i$ on the graph ${\mathcal G}$.
    \EndFor
\EndFor

%=====================================================================================
\end{algorithmic}
%=====================================================================================
\end{algorithm}

As can be seen from~\Cref{algo:1}, each agent $i$ updates its local parameter $\theta_t^i$ using its own reward $R_i^\pi$ and parameters of its neighbors $\theta_t^j,j\in {\cal N}_i$. Nevertheless, we can prove that each agent $i$ can find the global solution $\theta_\infty^c$ given in~\eqref{eq:theta-global}. To this end, we first provide stationary points of this system in the next result, and then prove that both weight vector $\bar \theta_t$ and auxiliary vector $\bar w_t$ reach the stationary point.
\begin{proposition}[Equilibrium points]~\label{proposition:proposition-ver1-equil}
The unique equilibrium point, $\bar \theta_\infty$, of the linear system in~\eqref{eq:ver1} corresponding to the vector $\bar \theta_t$ is given by $\bar \theta ^*  = {\bf{1}}_N  \otimes \theta _\infty ^c$, where $\theta_\infty ^c$ is defined in~\eqref{eq:theta-global}. Moreover, for the auxiliary vectors $\bar w_t$, the corresponding equilibrium points are all solutions, $\bar w_\infty$, of the following linear equation:
\begin{align}
\bar L\bar w_\infty = \begin{bmatrix}
   \Phi^T D\left(R_1^\pi - \frac{1}{N}\sum_{i=1}^N {R_i^\pi} \right)\\
   \Phi^T D\left(R_2^\pi - \frac{1}{N}\sum_{i=1}^N {R_i^\pi} \right)\\
    \vdots   \\
   \Phi^T D\left(R_N^\pi - \frac{1}{N}\sum_{i=1}^N {R_i^\pi} \right)\\
\end{bmatrix}\label{eq:lin_eq}.
\end{align}
\end{proposition}

The proof of~\cref{proposition:proposition-ver1-equil} is given in Appendix~\ref{proof:1}.
\cref{proposition:proposition-ver1-equil} implies that the local parameter $\theta^i_t$ reaches a consensus, i.e.,
\[
\lim_{t \to \infty } \theta _t^1  = \lim_{t \to \infty } \theta _t^2  =  \cdots  = \lim_{t \to \infty } \theta _t^N  = \theta _\infty ^c.  \]
On the other hand, $\bar w_\infty$ lies in an affine subspace, which is infinite. Next, we prove global asymptotic stability of the equilibrium points, whose proof is given in Appendix~\ref{proof:2}.
\begin{proposition}[Global asymptotic stability]~\label{proposition:proposition-ver1-GAS} The equilibrium points $(\bar \theta_\infty, \bar w_\infty)$ of~\eqref{eq:ver1} is globally asymptotically stable, i.e., $\bar \theta_t \to \bar \theta_\infty = \bar \theta^*$ and $\bar w_t \to \bar w_\infty$ as $t\to \infty$.
\end{proposition}

\cref{proposition:proposition-ver1-GAS} establishes that the first DP version converges to the solution $\bar \theta^*$.
As a potential application, it is easy to envision the development of a distributed RL by replacing certain terms in~\Cref{algo:1} with sample transitions of the underlying MDP. In such a scenario, the asymptotic stability of the continuous-time DP in~\Cref{algo:1} could be leveraged to demonstrate the convergence of the RL, using the well-established Borkar-Meyn theorem~\cite{borkar2000ode}. A primary challenge in applying the Borkar-Meyn theorem is the non-uniqueness of the equilibrium point of the ODE in~\eqref{eq:ver1}. For the Borkar-Meyn theorem's application, the existence of a unique equilibrium point is a prerequisite. As such, the stability analysis for~\Cref{algo:1} cannot be directly converted to its RL counterpart. In the following subsection, we introduce the second version, which can potentially address the aforementioned challenges.

\subsection{Distributed dynamic programming version~2}

Motivated by the aforementioned discussion, we propose the following continuous-time DP:
\begin{align}
\frac{d}{dt}\bar\theta_t  &= \bar\Phi^T \bar D(-I+\gamma\bar P^\pi) \bar\Phi\bar\theta_t+\bar\Phi^T \bar D\bar R^\pi - \bar L\bar \theta _t \nonumber\\
\frac{d}{dt}\bar w_t  &= \bar\theta_t-\bar w_t - \bar L\bar w_t -\bar L\bar v_t\nonumber\\
\frac{d}{dt}\bar v_t  &= \bar L\bar w_t\label{eq:dv/dt=Lw_t}
\end{align}
\begin{algorithm}[t]
%=====================================================================================
\caption{Distributed dynamic programming version~2}
\begin{algorithmic}[1]
%=====================================================================================

\State Initialize $\theta_0^i,w_0^i,i\in \{1,2,\ldots,N \}$.

\For{$t \ge 0$}
    \For{agent $i\in \{1,\ldots,N\}$}
    \State Update
    \begin{align*}
    \begin{split}
    \frac{d}{dt}\theta_t^i &= \Phi^T D(-I+\gamma P^\pi)\Phi\theta_t^i +\Phi^T DR_i^\pi \\ &-\left(|{\mathcal N}_i|\theta_t^i - \sum_{j\in {\mathcal N}_i} {\theta_t^j}\right)
    \end{split} \\
    \begin{split}
    \frac{d}{dt}w_t^i &= \theta_t^i - w_t^i\\ &- \left( |{\mathcal N}_i |w_t^i-\sum_{j\in {\mathcal N}_i}{w_t^j} \right) -\left(| {\mathcal N}_i|v_t^i -\sum_{j \in {\mathcal N}_i}{v_t^j}\right)
    \end{split}\\
    \begin{split}
    \frac{d}{dt}v_t^i {}&= |{ \mathcal N}_i|w_t^i - \sum_{j\in {\mathcal N}_i} {w_t^j}
    \end{split}
    \end{align*}
    where ${\mathcal N}_i$ is the neighborhood of node $i$ on the graph ${\mathcal G}$.
    \EndFor
\EndFor
%=====================================================================================
\end{algorithmic}\label{algo:2}
%=====================================================================================
\end{algorithm}

The overall algorithm, when viewed locally, is summarized in~\Cref{algo:2}.
We can demonstrate that $\bar w_t \to \bar \theta^*$ as $t\to \infty$, where $\bar \theta^*$ is defined in~\eqref{eq:theta-global2}. Thus, this DP can serve as an alternative to the DP in~\eqref{eq:ver1}. A key distinction between the current DP and its predecessor is the decoupling of the ODE corresponding to $\bar \theta_t$ from the components linked to $(\bar w_t, \bar v_t)$
The ODE for $\bar \theta_t$ can be seen as the local value function estimation, while the ODE for $(\bar w_t, \bar v_t)$ represents the parameter mixing component. This characteristic renders it more apt for RL applications. We will first establish the equilibrium points of~\eqref{eq:dv/dt=Lw_t} and their asymptotic stability.
\begin{proposition}[Equilibrium points]\label{proposition:3}
The unique equilibrium point, $\bar w_\infty$, of the linear system in~\eqref{eq:dv/dt=Lw_t} corresponding to the vector $\bar w_t$ is given by $\bar w_\infty = {\bf{1}}_N  \otimes \theta _\infty ^c = \bar \theta^*$, where $\theta_\infty ^c$ is defined in~\eqref{eq:theta-global}.
Moreover, for the vector $\bar \theta_t$, the corresponding equilibrium points, $\bar \theta_\infty$, are all solutions of the following linear equation:
\begin{align}
\frac{1}{N}\sum_{i=1}^N {\theta_\infty^i} = -(\Phi^T D(-I+\gamma P^\pi)\Phi)^{-1} \Phi^T D\left( \frac{1}{N}\sum_{i=1}^N {R_i^\pi}\right)\label{eq:3}
\end{align}
For another vector $\bar v_t$, the corresponding equilibrium points, $\bar v_\infty$, are all solutions of the following linear equation:
\begin{align}
\bar L\bar v_\infty =\bar\theta_\infty-\bar w_\infty.\label{eq:5}
\end{align}
\end{proposition}

The proof of~\cref{proposition:3} is given in Appendix~\ref{proof:3}.
Next, we establish the global asymptotic stability of~\eqref{eq:dv/dt=Lw_t}, whose proof is given in Appendix~\ref{proof:4}.
\begin{proposition}[Global asymptotic stability]\label{proposition:proposition-ver2-GAS}
The equilibrium points $(\bar \theta_\infty, \bar w_\infty, \bar v_\infty)$ of~\eqref{eq:dv/dt=Lw_t} is globally asymptotically stable, i.e., $\bar \theta_t \to \bar \theta_\infty$, $\bar w_t \to \bar w_\infty = \bar \theta^*$, and $\bar v_t \to \bar v_\infty$ as $t\to \infty$.
\end{proposition}

\begin{example}
Let us consider the Markov decision process borrowed from~\cite{lee2022distributed} with
\begin{align*}
&P^\pi =\begin{bmatrix}
   0.1 & 0.5 & 0.2 & 0.2 \\
   0.5 & 0.0 & 0.1 & 0.4 \\
   0.0 & 0.9 & 0.1 & 0.0 \\
   0.2 & 0.1 & 0.1 & 0.6  \\
\end{bmatrix},
\end{align*}
where $\pi$ is not explicitly specified, $|{\mathcal S}|=4$, $\gamma = 0.8$, and the local expected reward functions
\begin{align*}
&R_1^\pi=\begin{bmatrix}
   0 & 0 & 0 & 50\\
\end{bmatrix}^T,\quad R_2^\pi=\begin{bmatrix}
   0 & 0 & 0 & 0\\
\end{bmatrix}^T,\\
&R_3^\pi = \begin{bmatrix}
   0 & 0 & 0 & 0\\
\end{bmatrix}^T,\quad R_4^\pi= \begin{bmatrix}
   0 & 0 & 0 & 0\\
\end{bmatrix}^T,\\
&R_5^\pi= \begin{bmatrix}
   0 & 0 & 0 & 0\\
\end{bmatrix}^T,
\end{align*}

The feature matrix is
\begin{align*}
\Phi  = \left[ {\begin{array}{*{20}{c}}
1&0\\
1&0\\
0&2\\
0&2
\end{array}} \right]
\end{align*}

and the five RL agents over the network given in~\cref{fig:graph}.
\begin{figure}[h]
\centering\epsfig{figure=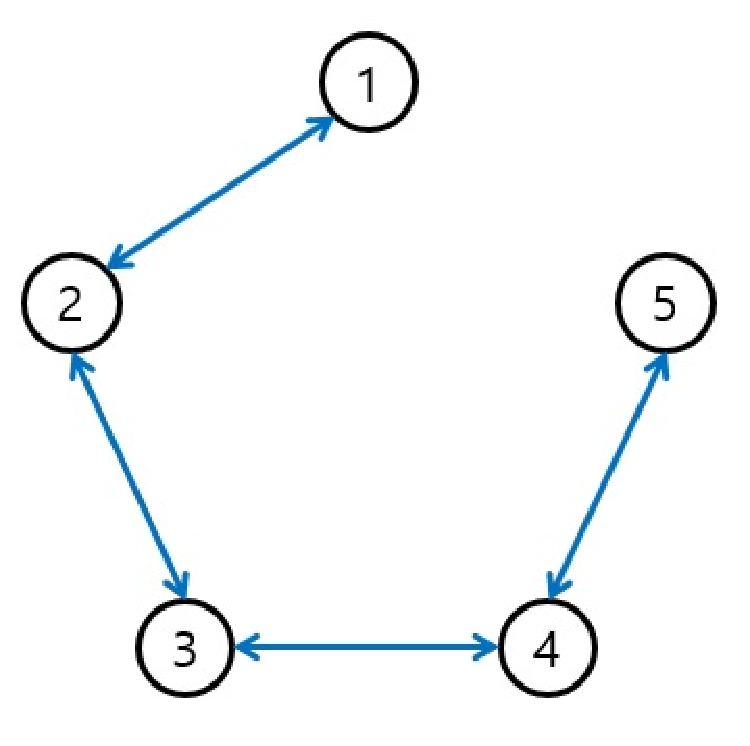,width=4cm} \caption{Network topology of five RL agents.}\label{fig:graph}
\end{figure}
\begin{figure}[h!]
\centering\epsfig{figure=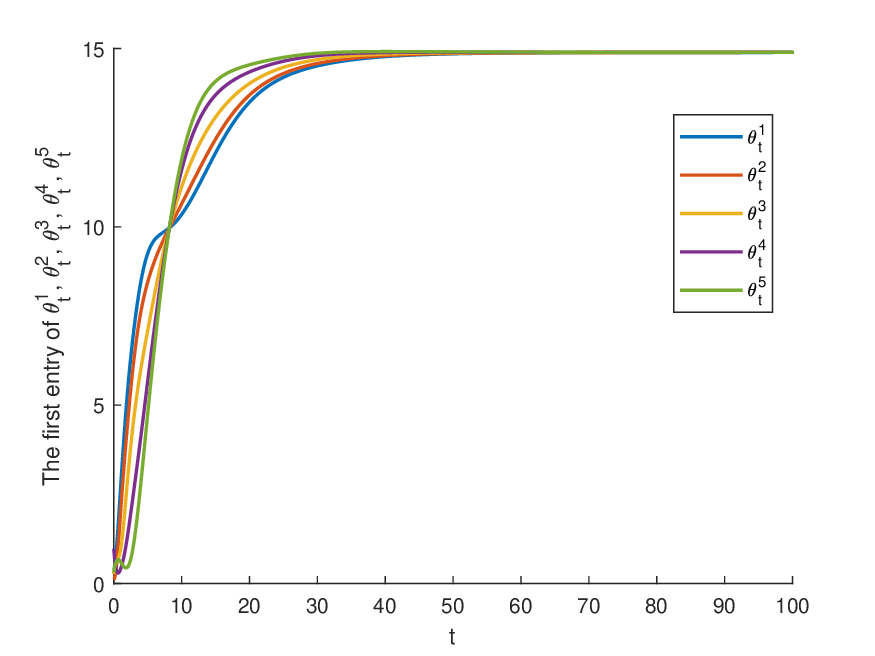,width=7cm}
\caption{\cref{algo:1}: Evolution of the first entries of $\theta_t^1$, $\theta_t^2$, $\theta_t^3$, $\theta_t^4$, and $\theta_t^5$.}\label{fig:1}
\end{figure}
\begin{figure}[h!]
\centering\epsfig{figure=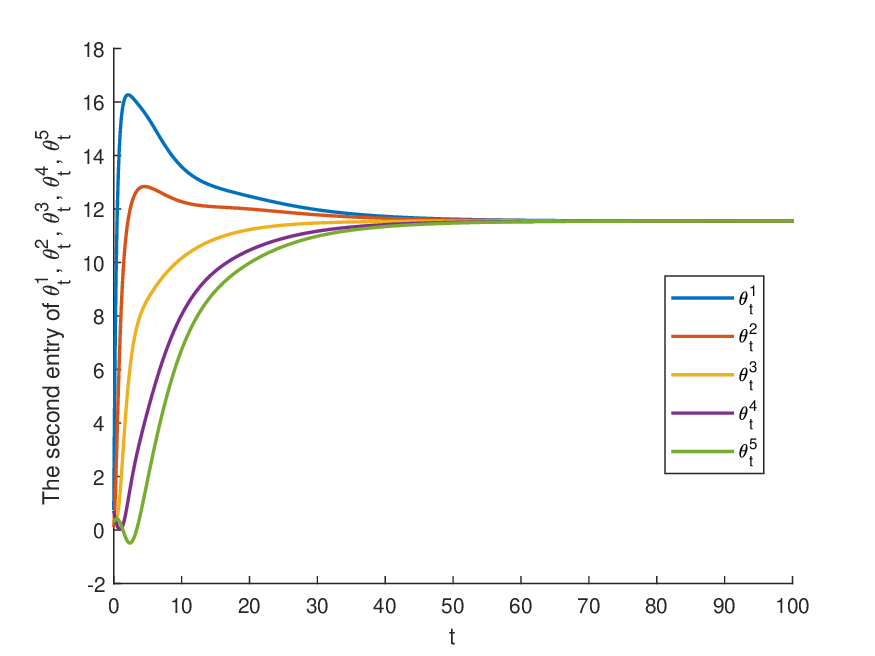,width=7cm}
\caption{\cref{algo:1}: Evolution of the second entries of $\theta_t^1$, $\theta_t^2$, $\theta_t^3$, $\theta_t^4$, and $\theta_t^5$.}\label{fig:2}
\end{figure}
\begin{figure}[h!]
\centering\epsfig{figure=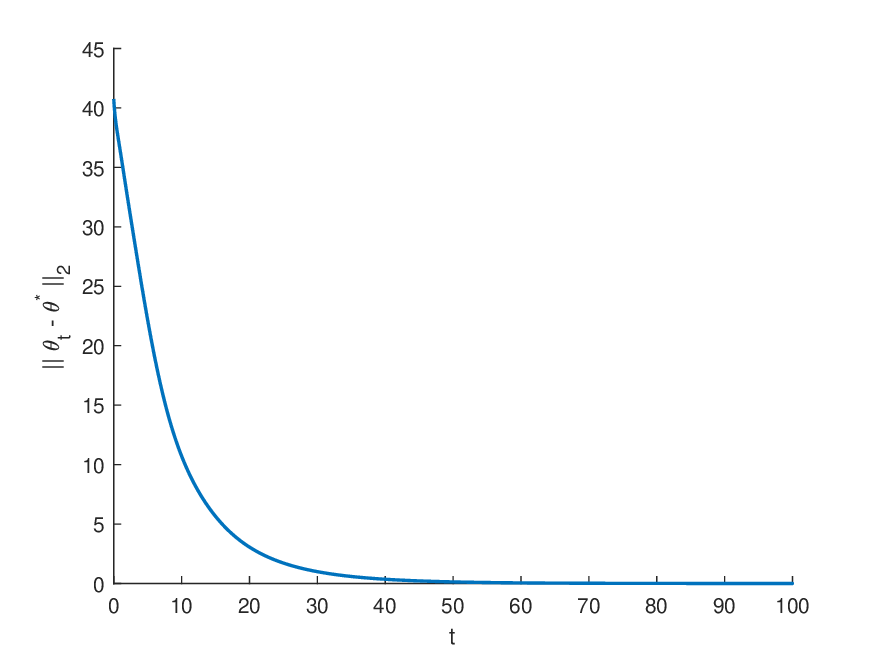,width=7cm}
\caption{\cref{algo:1}: Evolution of ${\left\| {{{\bar \theta }_t} - {{\bar \theta }^*}} \right\|_2}$.}\label{fig:3}
\end{figure}
For~\cref{algo:1},~\cref{fig:1} depicts the evolutions of the first entries of $\theta_t^1$, $\theta_t^2$, $\theta_t^3$, and $\theta_t^4$. Similarly,~\cref{fig:2} illustrates the evolutions of the first entries of $\theta_t^1$, $\theta_t^2$, $\theta_t^3$, and $\theta_t^4$. These results demonstrate that the parameters of the five agents reach a consensus.
\cref{fig:3} shows the evolution of the error $\| {{{\bar \theta }_t} - {{\bar \theta }^*}} \|_2$, and empirically proves that the parameter of the agents ${\bar \theta }_t$ converges to the optimal solution ${\bar \theta }^*$.

\begin{figure}[h!]
\centering\epsfig{figure=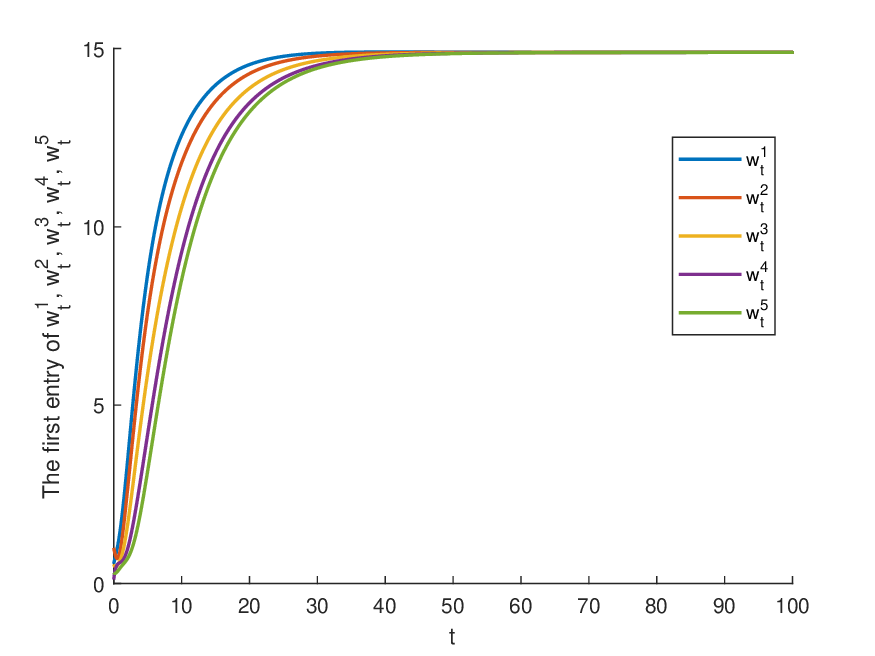,width=7cm}
\caption{\cref{algo:2}: Evolution of the first entries of $w_t^1$, $w_t^2$, $w_t^3$, $w_t^4$, and $w_t^5$.}\label{fig:4}
\end{figure}
\begin{figure}[h!]
\centering\epsfig{figure=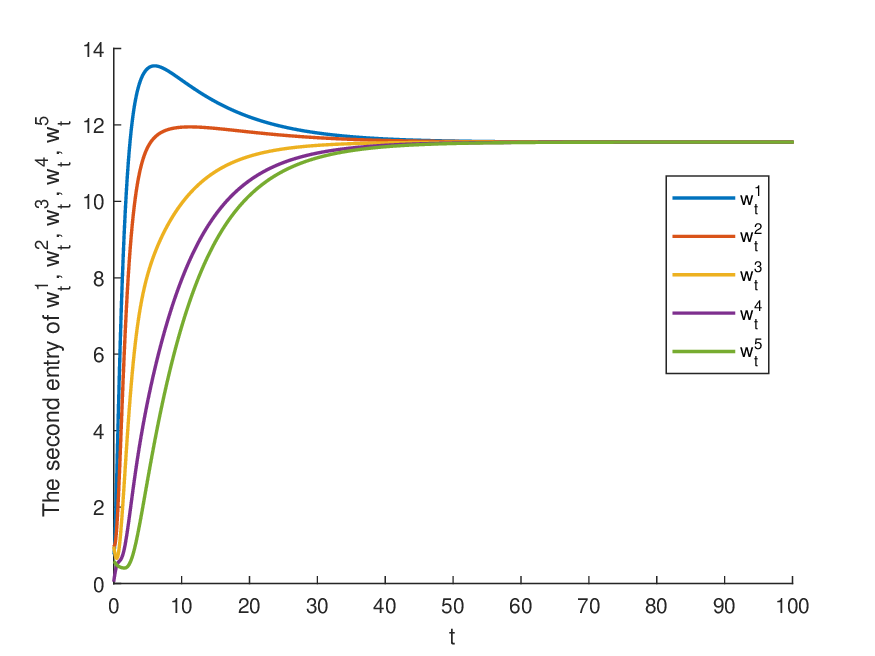,width=7cm}
\caption{\cref{algo:2}: Evolution of the second entries of $w_t^1$, $w_t^2$, $w_t^3$, $w_t^4$, and $w_t^5$.}\label{fig:5}
\end{figure}
\begin{figure}[h!]
\centering\epsfig{figure=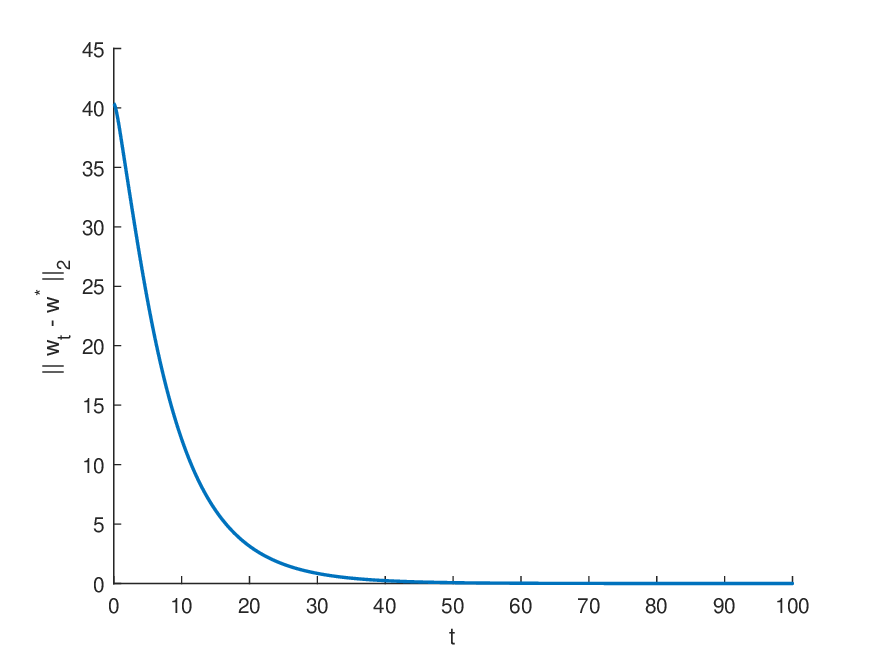,width=7cm}
\caption{\cref{algo:2}: Evolution of ${\left\| {{{\bar w}_t} - {{\bar \theta }^*}} \right\|_2}$.}\label{fig:6}
\end{figure}
Next,~\cref{fig:4},~\cref{fig:5},~\cref{fig:6} give similar results corresponding to~\cref{algo:2}.
The results also empirically demonstrate the validity of the proposed~\cref{algo:2}. We can also observe that both algorithms have similar convergence speeds.
\end{example}

\section{Conclusion}
This paper introduces new continuous-time distributed DP algorithms for MAMDPs and establishes their convergence. These are the initial efforts to develop distributed DP algorithms with simple continuous-time linear dynamics. The development and analysis, based on in systems and control theory, allow for more intuitive analysis from a control theory perspective, and improve clarity particularly for people with backgrounds in systems and control. The results in this paper offer further insights into distributed DP algorithms. Furthermore, the paper sets a foundation for the development of new distributed temporal difference learning algorithms. Finally, we expect that the results in this paper can be potentially extended to the reinforcement learning~\cite{sutton1998reinforcement} counterparts using the O.D.E. methods~\cite{kushner2003stochastic,borkar2000ode}. Moreover, the results can be extended to the multi-agent Q-learning scenarios to find optimal policies, for example, using the switching system framework in~\cite{lee2020unified,lee2021discrete}. which are potential future topics.

\bibliographystyle{IEEEtran}
\bibliography{reference}

\section{Appendix}

\subsection{Proof of Proposition~\ref{proposition:proposition-ver1-equil}}\label{proof:1}
Let $(\bar \theta_\infty, \bar w_\infty)$ be an equilibrium point corresponding to $(\bar \theta_t, \bar w_t)$. Then, it should satisfy the following equation:
\begin{align*}
\bar\Phi^T \bar D(-I+\gamma \bar P^\pi)\bar\Phi \bar\theta_\infty + \bar\Phi^T \bar D\bar R^\pi-\bar L \bar \theta_\infty-\bar L \bar w_\infty&=0 \\
\bar L \bar\theta_\infty&=0.
\end{align*}

The second equation implies $\bar \theta _\infty = {\bf{1}}_N  \otimes v$ for some $v \in {\mathbb R}^q$.
Plugging this relation into the first equation, we have
\begin{align}
\bar \Phi^T \bar D(-I + \gamma\bar P^\pi)\bar \Phi \bar \theta_\infty + \bar\Phi^T \bar D\bar R^\pi - \bar L \bar w_\infty = 0,\label{eq:proof-1}
\end{align}
where we used $\bar L \bar\theta_\infty=0$. Multiplying $({\bf 1}_N  \otimes I)^T$ from the left yields
\[
({\bf 1}_N \otimes I)^T [\bar\Phi^T \bar D(-I+\gamma \bar P^\pi)\bar\Phi \bar\theta_\infty + \bar\Phi^T \bar D\bar R^\pi] = 0,
\]
where we used $({\bf{1}}_N  \otimes I)^T \bar L \bar w_\infty = 0$. The equation can be equivalently written as
\[
N\Phi^T D(-I + \gamma P^\pi)\Phi v = -\sum_{i=1}^N {\Phi^T DR_i^\pi}
\]

Since $\Phi^T D(-I + \gamma P^\pi)\Phi$ is nonsingular from~\cite[pp.~209]{bhatnagar2012stochastic}, it follows that
\[
v = -(\Phi^T D(-I + \gamma P^\pi)\Phi)^{-1} \Phi ^T D\left( \frac{1}{N}\sum_{i=1}^N {R_i^\pi} \right),
\]
which is identical to~\eqref{eq:theta-global}, i.e., $v = \theta_\infty^c$. Therefore, $\bar \theta _\infty ^c  = {\bf{1}}_N  \otimes v = {\bf{1}}_N  \otimes \theta _\infty ^c   = \bar \theta ^*$. Plugging it back to~\eqref{eq:proof-1} leads to~\eqref{eq:lin_eq}. This completes the proof.

\subsection{Proof of Proposition~\ref{proposition:proposition-ver1-GAS}}\label{proof:2}
With $\bar\Phi^T \bar D(-I + \gamma\bar P^\pi)\bar\Phi \bar\theta_\infty + \bar\Phi^T \bar D\bar R^\pi - \bar L\bar w_\infty=0$ and $\bar L \bar \theta_\infty = 0$, the ODEs in~\eqref{eq:ver1} become
\begin{align}
\begin{split}
\frac{d}{dt}(\bar\theta_t - \bar\theta_\infty) {}&= \bar\Phi^T \bar D(-I + \gamma\bar P^\pi)\bar\Phi (\bar\theta_t - \bar\theta_\infty)\\
& -\bar L(\bar\theta_t-\bar\theta_\infty) - \bar L(\bar w_t - \bar w_\infty)\label{eq:ode-theta}
\end{split}\\
\begin{split}
\frac{d}{dt}(\bar w_t - \bar w_\infty) {}&= \bar L(\bar\theta_t - \bar\theta_\infty)\nonumber.
\end{split}
\end{align}
Consider the function
\[
V\left( \begin{bmatrix}
   \bar\theta_t - \bar\theta_\infty \\
   \bar w_t  - \bar w_\infty\\
\end{bmatrix} \right) = \begin{bmatrix}
   \bar \theta_t  - \bar \theta_\infty\\
   \bar w_t - \bar w_\infty  \\
\end{bmatrix}^T \begin{bmatrix}
   \bar \theta_t  - \bar\theta_\infty\\
   \bar w_t  - \bar w_\infty \\
\end{bmatrix}.
\]
Its time-derivative along the trajectory is
\begin{align*}
&\dot V\left( \begin{bmatrix}
   \bar\theta_t - \bar \theta_\infty\\
   \bar w_t - \bar w_\infty \\
\end{bmatrix} \right)\\
=& \begin{bmatrix}
   \bar\theta_t -\bar\theta_\infty\\
   \bar w_t -\bar w_\infty\\
\end{bmatrix}^T \begin{bmatrix}
   2\bar\Phi^T \bar D(-I + \gamma\bar P^\pi)\bar\Phi - 2\bar L & 0\\
   0 & 0  \\
\end{bmatrix}\\
&\times\begin{bmatrix}
   \bar\theta_t -\bar\theta_\infty\\
   \bar w_t - \bar w_\infty  \\
\end{bmatrix}\\
=& \ (\bar\theta_t - \bar\theta_\infty)^T [2\bar\Phi^T \bar D(-I + \gamma\bar P^\pi)\bar\Phi - 2\bar L](\bar\theta_t-\bar\theta_\infty).
\end{align*}
Using the following well-known inequality~\cite[pp.~209]{bhatnagar2012stochastic}:
\begin{align}
\bar \Phi^T\bar D(\gamma\bar P^\pi-I)\bar\Phi +\bar\Phi ^T (\gamma \bar P^\pi-I)^T \bar D\bar \Phi  \preceq 2(\gamma  - 1)\bar D,\label{eq:8}
\end{align}
one gets $\bar \Phi ^T \bar D( - I + \gamma \bar P^\pi  )\bar \Phi  + \bar \Phi ^T (\gamma \bar P^\pi   - I)^T \bar D\bar \Phi  -  2 \bar L \preceq 2(\gamma  - 1) \bar D - 2 \bar L \prec 0$, where the second inequality is due to $\bar D \succ 0, \bar L \succeq 0$, and $\gamma -1 <0$.
Equivalently, we have
\begin{align*}
&\dot V\left(\begin{bmatrix}
   \bar\theta_t - \bar\theta_\infty\\
   \bar w_t - \bar w_\infty  \\
\end{bmatrix}\right)\\
\le& (\bar\theta_t - \bar\theta_\infty)^T ((\gamma-1)\bar D - \bar L)(\bar\theta_t -\bar\theta_\infty)\\
<& 0,
\end{align*}
for any $\bar\theta_t -\bar\theta_\infty \neq 0$.
Taking the integral on both sides and rearranging terms lead to
\begin{align*}
&V\left(\begin{bmatrix}
   \bar \theta_T - \bar \theta_\infty\\
   \bar w_T - \bar w_\infty \\
\end{bmatrix}\right) - V\left(\begin{bmatrix}
   \bar\theta_0 - \bar\theta_\infty\\
   \bar w_0 - \bar w_\infty\\
\end{bmatrix}\right)\\
\le & \int_0^T {(\bar\theta_t - \bar\theta_\infty)^T ((\gamma-1)\bar D - \bar L)(\bar\theta_t - \bar\theta_\infty)dt}
\end{align*}
Rearranging terms lead to
\begin{align*}
&\lambda _{\min } ((1 - \gamma )\bar D + \bar L)\int_0^T {(\bar \theta _t  - \bar \theta _\infty  )^T (\bar \theta _t  - \bar \theta _\infty  )dt}\\
\le&\int_0^T {(\bar \theta _t  - \bar \theta _\infty  )^T ((1 - \gamma )\bar D +\bar L)(\bar \theta _t  - \bar \theta _\infty  )dt} \\
\le& V\left( \begin{bmatrix}
   {\bar \theta _0  - \bar \theta _\infty  }  \\
   {\bar w_0  - \bar w_\infty  }  \\
\end{bmatrix} \right) - V\left( \begin{bmatrix}
   {\bar \theta _T  - \bar \theta _\infty  }  \\
   {\bar w_T  - \bar w_\infty  }  \\
\end{bmatrix} \right)\\
\le& V\left( \begin{bmatrix}
   {\bar \theta _0  - \bar \theta _\infty  }  \\
   {\bar w_0  - \bar w_\infty  }  \\
\end{bmatrix} \right).
\end{align*}

Taking the limit $T \to \infty$ yields
\begin{align*}
&\int_0^\infty  {(\bar\theta_t  - \bar\theta_\infty)^T (\bar\theta_t - \bar\theta_\infty)dt}\\
\le& \frac{1}{\lambda_{\min}((\gamma-1)\bar D-\bar L)}
V\left( \begin{bmatrix}
\bar\theta_0 - \bar\theta_\infty\\
\bar w_0 - \bar w_\infty\\
\end{bmatrix}\right)
\end{align*}
which implies from Barbalat's lemma, that $\bar\theta_t \to \bar\theta_\infty$ as $t \to \infty$. Now, taking the limit $t \to \infty$ on both sides of~\eqref{eq:ode-theta} yields $\bar L\bar w_\infty = \lim_{t\to\infty} \bar L\bar w_t$.
Combining the above equation with~\eqref{eq:lin_eq} leads to $\lim_{t\to\infty} \bar L\bar w_t = 0$, which is the desired conclusion.

\subsection{Proof of Proposition~\ref{proposition:3}}\label{proof:3}
First of all, note that the stationary points should satisfy
\begin{align}
\bar\Phi^T\bar D(-I+\gamma\bar P^\pi)\bar\Phi\bar\theta_\infty +\bar\Phi^T \bar D\bar R^\pi-\bar L\bar\theta_\infty &=0\nonumber\\
\bar\theta_\infty -\bar w_\infty - \bar L\bar w_\infty-\bar L\bar v_\infty &=0\nonumber\\
\bar L\bar w_\infty &=0\label{eq:6}
\end{align}
Multiplying $({\bf 1}_N \otimes I)^T$ from the left, the first equation becomes
\[
\Phi^T D(-I+\gamma P^\pi)\Phi \sum_{i=1}^N {\theta_\infty^i}+\Phi^T D\sum_{i=1}^N {R_i^\pi}= 0.
\]
Rearranging terms, we can prove that $\theta_\infty^i$ should satisfy~\eqref{eq:3}. On the other hand, the third equation implies
\begin{align}
w_\infty^1 = w_\infty^2 = \cdots = w_\infty ^N =:w_\infty.\label{eq:4}
\end{align}
Plugging this relation into the second equation and multiplying $({\bf 1}_N\otimes I)^T$ from the left, we have
\[
({\bf 1}_N \otimes I)^T \bar\theta_\infty-({\bf 1}_N\otimes I)^T \bar w_\infty=0.
\]
Combining the above equation with~\eqref{eq:4} leads to
\begin{align*}
w_\infty &= \frac{1}{N}\sum_{i=1}^N {\theta_\infty^i}\\&= -(\Phi^T D(-I+\gamma P^\pi)\Phi)^{-1} \Phi^T D\left( \frac{1}{N}\sum_{i=1}^N {R_i^\pi}\right),
\end{align*}
where the second equality comes from~\eqref{eq:3}. Finally, the second equation with $\bar L\bar w_\infty$ results in~\eqref{eq:5}. This completes the proof.

\subsection{Proof of Proposition~\ref{proposition:proposition-ver2-GAS}}\label{proof:4}
Noting that the equilibrium points satisfy~\eqref{eq:6}, the ODEs in~\eqref{eq:dv/dt=Lw_t} can be written by
\begin{align}
\frac{d}{dt}(\bar\theta_t -\bar\theta_\infty) &=[\bar\Phi^T \bar D(-I+\gamma\bar P^\pi)\bar\Phi-\bar L](\bar\theta_t - \bar \theta _\infty  )\nonumber\\
\frac{d}{dt}(\bar w_t-\bar w_\infty)  &= (\bar\theta_t-\bar\theta_\infty)-(I+\bar L)(\bar w_t-\bar w_\infty)\nonumber\\
&-\bar L(\bar v_t - \bar v_\infty)\nonumber\\
\frac{d}{dt}(\bar v_t-\bar v_\infty) &=\bar L(\bar w_t - \bar w_\infty)\label{eq:9}
\end{align}
Let us consider the Lyapunov function candidate
\[
V(\bar\theta_t -\bar\theta_\infty) = (\bar\theta_t-\bar\theta_\infty)^T (\bar\theta_t-\bar\theta_\infty),
\]
whose time-derivative along the trajectory is
\begin{align*}
&\frac{d}{dt}V(\bar\theta_t-\bar\theta_\infty)\\
=& \ (\bar\theta_t-\bar\theta_\infty)^T (\bar D(-I+\gamma \bar P^\pi)\bar\Phi \\
&+ (\bar D(-I+\gamma\bar P^\pi)\bar\Phi)^T - 2\bar L)(\bar\theta_t - \bar\theta_\infty)\\
<& \ 0
\end{align*}
for all $\bar\theta_t -\bar\theta_\infty\neq 0$, where the last inequality is due to~\eqref{eq:8} and $\bar L \succeq 0$. By the Lyapunov theorem~\cite{khalil2002nonlinear}, $\bar\theta_t \to \bar\theta_\infty$ as $t \to \infty$. Moreover, since the system is a linear system, the convergence is exponential. For the convergence of $\bar w_t$, consider the function
\begin{align*}
&V(\bar w_t-\bar w_\infty,\bar v_t-\bar v_\infty)\\
=& \ (\bar w_t -\bar w_\infty)^T (\bar w_t-\bar w_\infty)+ (\bar v_t - \bar v_\infty)^T (\bar v_t -\bar v_\infty),
\end{align*}
whose time-derivative along the trajectories is
\begin{align*}
&\frac{d}{dt}V(\bar w_t - \bar w_\infty,\bar v_t -\bar v_\infty) \\
=& -(\bar w_t -\bar w_\infty)^T (2I + 2\bar L)(\bar w_t-\bar w_\infty)\\
&+2(\bar w_t -\bar w_\infty)^T (\bar\theta_t -\bar\theta_\infty).
\end{align*}

Integrating both sides from $t=0$ to $T$ yields
\begin{align*}
V(&\bar w_T - \bar w_\infty, \bar v_T-\bar v_\infty)-V(\bar w_0 - \bar w_\infty,\bar v_0 -\bar v_\infty) \\&= -\int_0^T {(\bar w_t - \bar w_\infty)^T (2I + 2\bar L)(\bar w_t -\bar w_\infty)dt}\\
&+2\int_0^T {(\bar w_t -\bar w_\infty)^T (\bar\theta_t-\bar\theta_\infty)dt}
\end{align*}
Rearranging terms lead to
\begin{align*}
&2\int_0^T {(\bar w_t-\bar w_\infty)^T (I + \bar L)(\bar w_t-\bar w_\infty)dt}\\
 &=-V(\bar w_T-\bar w_\infty,\bar v_T - \bar v_\infty) + V(\bar w_0 - \bar w_\infty,\bar v_0 - \bar v_\infty)\\
&+2\int_0^T {(\bar w_t - \bar w_\infty)^T (\bar\theta_t-\bar\theta_\infty)dt}\\
 &\le V(\bar w_0 - \bar w_\infty,\bar v_0 - \bar v_\infty)+2\int_0^T {(\bar w_t -\bar w_\infty)^T (\bar \theta_t - \bar\theta_\infty)dt}\\
\begin{split}
&\le V(\bar w_0 - \bar w_\infty,\bar v_0 - \bar v_\infty) + \int_0^T {(\bar w_t-\bar w_\infty)^T (\bar w_t  - \bar w_\infty)dt} \\&+ \int_0^T {(\bar\theta_t-\bar\theta_\infty)^T (\bar\theta_t-\bar\theta_\infty)dt}
\end{split}
\end{align*}
where the second inequality comes from the Young's inequality. Rearranging some terms again, we have
\begin{align*}
&\int_0^T {(\bar w_t - \bar w_\infty)^T (I+2\bar L)(\bar w_t - \bar w_\infty)dt}\\
\le& \ V(\bar w_0-\bar w_\infty,\bar v_0 - \bar v_\infty)+\int_0^T {(\bar\theta_t -\bar\theta_\infty)^T (\bar\theta_t-\bar \theta_\infty)dt}
\end{align*}
The integral on the right-hand side is bounded because $\bar\theta _t  - \bar \theta _\infty$ converges to zero exponentially. Moreover, since $I + 2\bar L$ is positive definite, the above inequality implies that $\bar w_t \to \bar w_\infty = \bar \theta^* =  {\bf{1}}_N  \otimes \bar \theta _\infty $ as $t\to \infty$ from the Barbalat's lemma. Now, taking the limit $t\to\infty$ on both sides of the third equation in~\eqref{eq:9} leads to
\[
\lim_{t\to\infty} \frac{d}{dt}(\bar v_t-\bar v_\infty)=0
\]
implying that $\bar v_t$ converges to some constant $\bar v_\infty$, where we used the fact that $\bar w_t \to \bar w_\infty$ as $t\to \infty$. Finally, it remains to prove the convergence of $\bar \theta_t$. To this end, taking the limit $t\to \infty$ on both sides of the second equation in~\eqref{eq:dv/dt=Lw_t} leads to $0 = \bar\theta_\infty - \bar w_\infty - \bar L\bar w_\infty - \lim_{t\to\infty} \bar L\bar v_t$, which is equivalent to $\lim_{t \to \infty } \bar L\bar v_t  = \bar \theta _\infty - \bar w_\infty $ using $\bar L\bar w_\infty = 0$. This completes the proof.

\end{document}